\documentclass[conference]{IEEEtran}
\RequirePackage[hyphens]{url}
\IEEEoverridecommandlockouts
\usepackage{graphicx,cite}
\usepackage{textcomp}
\def\BibTeX{{\rm B\kern-.05em{\sc i\kern-.025em b}\kern-.08em
    T\kern-.1667em\lower.7ex\hbox{E}\kern-.125emX}}
\usepackage[utf8]{inputenc}
\usepackage{amsmath, amssymb, mathtools, amsthm}
\usepackage{array}
\usepackage{booktabs} 
\usepackage{multicol}
\usepackage[shortlabels]{enumitem}
\usepackage{siunitx}
\usepackage{cancel}
\usepackage{graphicx}
\usepackage{listings}
\usepackage{bm}
\usepackage{subfigure}
\newtheorem{theorem}{Theorem}
\newtheorem{lemma}{Lemma}

\allowdisplaybreaks

\usepackage[colorlinks=true, allcolors=blue]{hyperref}
\hypersetup{breaklinks=true}

\usepackage{orcidlink}

\title{Co-Optimizing Distributed Energy Resources under Demand Charges and Bi-Directional Power Flow}

\makeatletter

\newcommand{\linebreakand}{%
      \end{@IEEEauthorhalign}
      \hfill\mbox{}\par
      \mbox{}\hfill\begin{@IEEEauthorhalign}
    }
    \makeatother
\author{%
\IEEEauthorblockN{Ruixiao Yang\textsuperscript{*}\orcidlink{0000-0002-2852-0150}, Gulai Shen\textsuperscript{*}\orcidlink{0000-0002-6016-0159}, Ahmed S. Alahmed\orcidlink{0000-0002-4715-4379}, Chuchu Fan\orcidlink{0000-0003-4671-233X}}

\thanks{\textsuperscript{*}Equal contribution.}
}

\date{}
\begin{document}
\maketitle
\begin{abstract}
    We address the co-optimization of behind-the-meter (BTM) distributed energy resources (DER), including flexible demands, renewable distributed generation (DG), and battery energy storage systems (BESS) under net energy metering (NEM) frameworks with demand charges. We formulate the problem as a stochastic dynamic program that accounts for renewable generation uncertainty and operational surplus maximization. Our theoretical analysis reveals that the optimal policy follows a threshold structure. Finally, we show that even a simple algorithm leveraging this threshold structure performs well in simulation, emphasizing its importance in developing near-optimal algorithms. These findings provide crucial insights for implementing prosumer energy management systems under complex tariff structures.
\end{abstract}
\section{Background and Literature}
A clean, resilient, and sustainable energy system requires optimal coordination of  BTM DER, such as DG, flexible demands, and BESS \cite{akorede_distributed_2010}. Economic and technical aspects are the main drivers for BTM DER operation and coordination. Under NEM, {\em prosumers} can bi-directionally transact energy and money with the distribution utility (DU), who levy a {\em buy} (retail) rate or a {\em sell} (export) rate based on net-consumption -- gross demand minus gross generation  \cite{Alahmed&Tong:22IEEETSG}. 

With the increasing penetration of new and large electric loads such as electric vehicles, electric ovens, and heat pumps, DUs are increasingly supplementing volumetric charges\footnote{Charges based on the total net consumption of energy.} with {\em demand charges}. These are defined as \$/kW fees imposed on the peak demand recorded during a billing cycle and are introduced to encourage smoother load profiles and reduce the operational costs of grid maintenance \cite{jin_optimal_2017}.\footnote{Demand charges differ from {\em critical peak pricing}, a form of time-varying pricing aimed at incentivizing load shedding during critical periods \cite{Neufeld:87JEH}.} With the growing adoption of demand charges, several studies have highlighted the role of BTM DER in helping customers mitigate them \cite{StorageDemandCharge:17NREL}.

This work focuses on the co-optimization of BTM DER, including flexible demand, renewable DG, and BESS, under NEM frameworks with differentiated buy and sell rates and demand charges (Fig. \ref{fig:der_setup}). A comparative overview of related studies in terms of demand flexibility, storage integration, NEM capabilities, and demand charges is provided in Table \ref{tab:comparison}. To the best of our knowledge, this work is the first to co-optimize flexible demand and storage operations under the combined presence of NEM and demand charges.

Previous research has addressed aspects of this problem. In \cite{Harshah&Dahleh:15TPS,Zhang&Leibowicz&Hanasusanto:20TSG}, dynamic programming was employed to minimize prosumer costs by optimally scheduling BESS operation under inelastic demand. Also, the co-optimization of demand flexibility and BESS was considered in \cite{Chen&Wang&Heo&Kishore:13TSG,Khezri&Mahmoudi&Haque:21TSE}, though without incorporating the ability of prosumers to buy and sell energy under NEM. Studies such as \cite{Li&Dong:19TSG,Xu&Tong:17TAC,alahmed_co-optimizing_2024,Jeddi&Mishra&Ledwich:20TSE} incorporate the NEM buy-and-sell feature, optimizing storage scheduling in \cite{Li&Dong:19TSG}, and jointly optimizing storage and flexible demand in \cite{Xu&Tong:17TAC,alahmed_co-optimizing_2024,Jeddi&Mishra&Ledwich:20TSE}. However, these works ignore demand charges, which significantly complicates the scheduling problem.

In contrast, while \cite{jin_optimal_2017, Luo&King&Ranzi&Dong:20TSG} integrate demand charges in their optimization models, the assumed free disposal of surplus renewables eliminates NEM's bidirectional energy transactions, which is a critical limitation given the growing real-world prevalence of such interactions.

\textbf{Contributions}. This paper advances research on optimal scheduling of BTM DERs with three key contributions:
\begin{enumerate}
    \item We develop a stochastic dynamic programming framework that jointly optimizes demand scheduling and BESS under NEM and demand charge constraints, maximizing prosumer operational surplus.
    \item We prove that the optimal policy maintains a threshold structure consistent with prior work while demonstrating how integrating multiple factors—demand flexibility, storage, and demand charges—leads to exponential computational complexity in the number of flexible demands.
    \item Through extensive simulations against baseline approaches, we demonstrate how leveraging the threshold structure properties enables the development of computationally efficient near-optimal algorithms.
\end{enumerate}



\begin{table}
    \centering
    \renewcommand{\arraystretch}{1.3} 
    \caption{Sample of Related work on home energy management.}
    \label{tab:comparison}
    \resizebox{\columnwidth}{!}{ 
    \begin{tabular}{@{}>{\centering\arraybackslash}m{2cm}cccc@{}}
        \toprule
        \textbf{Related work} & \textbf{Flexible demands} & \textbf{Storage} & \textbf{NEM (Buy/Sell)} & \textbf{Demand charge} \\ 
        \midrule
\cite{Harshah&Dahleh:15TPS,Zhang&Leibowicz&Hanasusanto:20TSG} & & \checkmark &&\\
\cite{Chen&Wang&Heo&Kishore:13TSG,Khezri&Mahmoudi&Haque:21TSE} & \checkmark & \checkmark &  &  \\ 
        \cite{Li&Dong:19TSG} &  & \checkmark & \checkmark & \\ 
        \cite{jin_optimal_2017} &  & \checkmark &  & \checkmark \\ 
        \cite{Luo&King&Ranzi&Dong:20TSG}& \checkmark & \checkmark &  & \checkmark\\ 
        \cite{Xu&Tong:17TAC,alahmed_co-optimizing_2024, Jeddi&Mishra&Ledwich:20TSE}& \checkmark & \checkmark & \checkmark & \\ 
        \textbf{This work} & \checkmark & \checkmark & \checkmark & \checkmark \\ 
        \bottomrule
    \end{tabular}
    }
\end{table}


\section{Problem Formulation}
\label{sec:problem}
\begin{figure}
    \centering
    \includegraphics[scale=0.27]{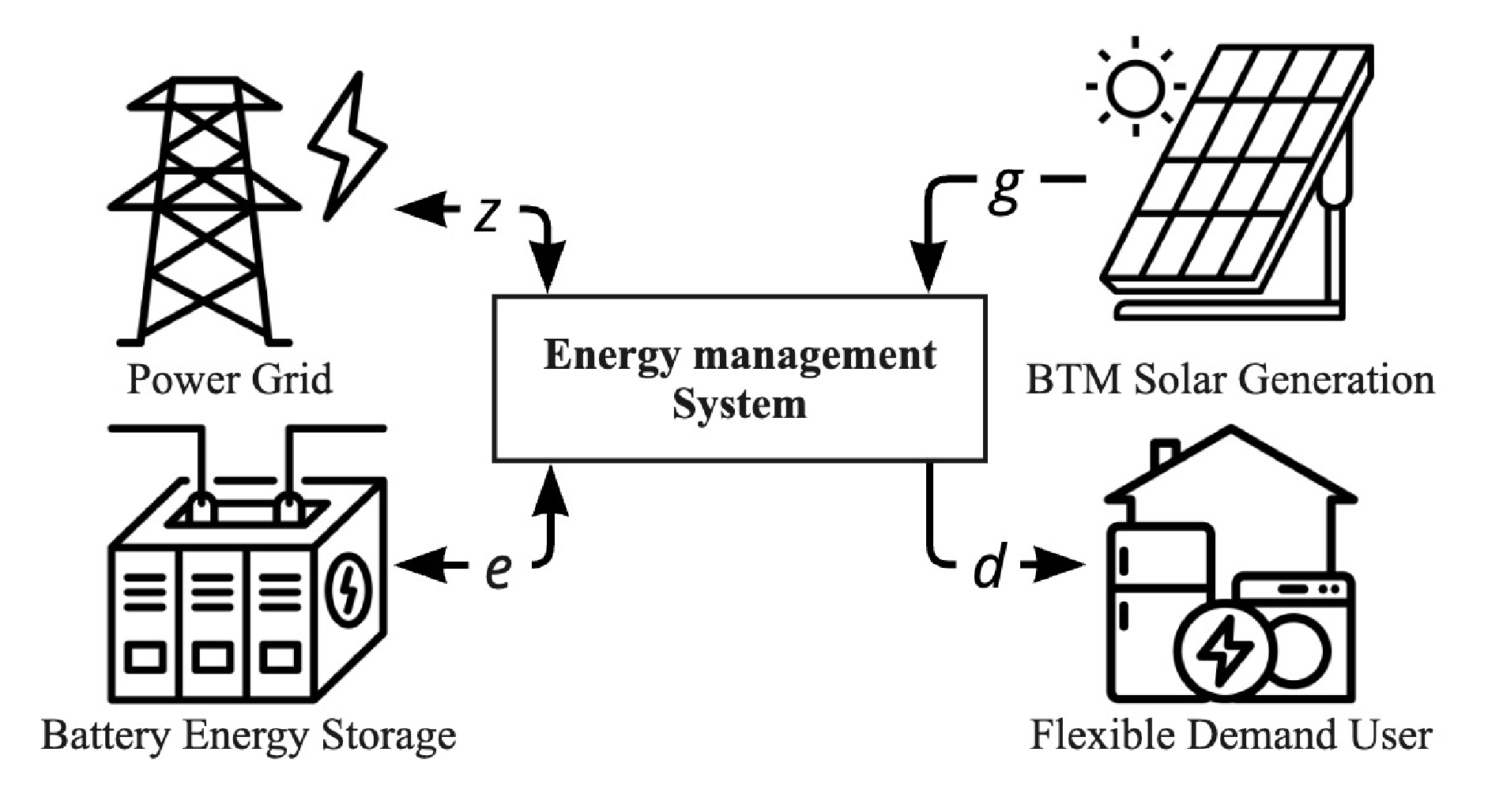}
    \caption{Distributed energy resource setup: flexible demands $d\in \mathbb R_+$, renewable DG $g\in\mathbb R_+$, storage operation $e\in \mathbb R$, and net-consumption $z\in \mathbb R$.}
    \label{fig:der_setup}
\end{figure}

We study an energy management system (EMS) for prosumers that co-optimizes flexible demand, storage operation, and renewable distributed generation under a NEM tariff with distinct buy/sell rates and peak demand charges. The optimization objective maximizes the net benefit defined as consumption utility minus electricity costs, where demand charges are levied on the maximum consumption drawn from either load, battery charging, or both. As shown in Fig.~\ref{fig:der_setup}, this dynamic programming formulation explicitly accounts for the temporal coupling of storage operations and an NEM tariff with bidirectional energy pricing and peak demand charges.

\subsection{Prosumer Resources}
We consider the operation over a discrete and finite time horizon $t=0,1,\ldots, T$. As shown in Fig.\ref{fig:der_setup}, the prosumer has {\em BTM renewable DG}, denoted by $g_t$. The renewable DG is modeled as an exogenous (positive) Markov random process.

The {\em flexible demands} are denoted by the vector 
\begin{equation}\boldsymbol{d}_t=(d_t^1, \ldots, d_t^K)\in\mathcal{D}\coloneqq\{\boldsymbol{d}:\boldsymbol{0}\preceq \boldsymbol{d}\preceq\overline{\boldsymbol{d}}\}\subseteq\mathbb R^K_+,
\end{equation}
where $K$ is the total number of prosumer devices. Here $\overline{\boldsymbol{d}}$ is the consumption bundle’s upper limit.

The prosumer also owns a BTM battery energy storage, which has output $e_t\in[-\underline{e}, \overline{e}]$, where $\underline{e}$ and $\overline{e}$ are discharging and charging limits. The output is decomposed into charging $[e_t]^+=\max\{e_t, 0\}$ and discharging $[e_t]^-=\max\{-e_t, 0\}$. With charge/discharge efficiencies $\tau,\rho\in[0, 1]$, the state of charge (SoC) $s_t\in[0, B]$ evolves as: 
\begin{equation}\label{eq:SoC}
    s_{t+1} = s_t + \tau [e_t]^+-[e_t]^-/\rho,\quad t=0,\ldots, T-1,
\end{equation} 
where $B$ denotes battery capacity.

The {\em net energy consumption} of the prosumer at time $t$ is
\begin{equation}
    z_t := \bm{1}^\top\boldsymbol{d}_t + e_t - g_t.
\end{equation} 
Therefore, the prosumer imports power from the grid if $z_t>0$ and exports power to the grid if $z_t<0$.

\subsection{System Dynamics}
The {\em system state} at time $t$, $x_t$, is described as 
\begin{equation}
    x_t = [s_t, g_t, c_t],
\end{equation}
where $c_t$ is the prosumer's peak demand before time $t$. We allow the system to be controlled by battery charging/discharging and consumption adjustments, i.e., the control action is
\begin{equation}
    u_t \coloneqq (e_t, \boldsymbol{d}_t).
\end{equation} 

 The peak demand in the state evolves as
\begin{equation}
    c_{t+1} = \max_{0\le i\le t}z_i=\max\{z_t,c_t\}, \quad t=0,1,\ldots, T-1,
\end{equation}
and the battery SoC evolves as per (\ref{eq:SoC}).

\subsection{NEM X Tariff with Demand Charges}
Under NEM, when the prosumer's gross loads (from consumption and storage charging) are higher than its generation (from renewable DG), the prosumer pays at the retail rate. Otherwise, the prosumer gets compensated at the export rate.

Let $(p_t^+\ge p_t^-, p) \ge 0$ be the retail, export, and peak demand prices, respectively. We equivalently break down the peak demand charge at the end of time $T$ to each time step $t$ since $p c_T = p\sum_{t=0}^{T-1}[z_t-c_t]^+$. Thus, we can write the payment (cost) under NEM X in \cite{Alahmed&Tong:22IEEETSG} for each time $t$ as 
\begin{equation}\label{eq:payment}
    P_t(z_t, c_t) = p_t^+[z_t]^+ - p_t^-[z_t]^- + p[z_t-c_t]^+,
\end{equation}
where $P_t(\cdot)>0$ ($P_t(\cdot)<0$) signifies a payment by the prosumer (DU) to the DU (prosumer).

\subsection{Flexible Demands and Storage Co-Optimization}
At each time step $t$, the prosumer's utility of consuming $\boldsymbol{d}_t$ is given by an additive utility function
\begin{equation}\label{eq:utility}
    U_t(\boldsymbol{d}_t) = \sum_{i=1}^K U_t^i(d_t^i), \quad t=0,1,\ldots, T-1
\end{equation}
that is assumed to be concave and continuously differentiable.

We define the reward function of the dynamic program as the difference between satisfaction from consumption (\ref{eq:utility}) and payment (\ref{eq:payment}). Therefore, the reward function is
\begin{equation}
    r_t(x_t, u_t) = U_t(\boldsymbol{d}_t) - P_t(z_t, c_t), \quad t=0,\ldots, T-1.
\end{equation}
To encourage the preservation of energy in the battery, we define a linear reward to the remaining energy at step $T$, {\em i.e. }
\begin{equation}
    r_T(x_T) = \gamma s_T,
\end{equation}
for a $\gamma>0$. The salvage value rate $\gamma$ influences optimal DER scheduling. Here, we assume $p^+_t \geq \gamma \geq p^-_t, \forall t$, deferring analysis of $\gamma$'s impact to future work.

The policy $\pi=(\pi_0, \ldots, \pi_{T-1})$ is a sequence of mapping from state space to action space such that $\pi_t(x_t) = u_t$. Given initial state $x_0=(s, g, 0)$, the storage-consumption co-optimization is then defined as:
\begin{subequations}
\begin{align}
    \max_{\pi}\quad& \mathbb E\left[\sum_{t=0}^{T}r_t(x_t, u_t)\right]\\
    \text{Subject to}&\quad \text{for all}\ t=0, 1,\cdots, T-1\nonumber\\
    & x_0 = (s, g, 0)\\
    &u_t=(e_t, \boldsymbol{d}_t)=\pi_t(x_t)\\
    & e_t\in[-\underline{e}, \overline{e}] \\
    & s_{t+1} = s_t + \tau[e_t]^+-[e_t]^-/\rho \label{SoCevolution}\\
    & 0 \leq s_t \leq B\\
    & g_{t+1}\sim F_{\cdot |g_t} \\
    & c_{t+1}=\max\{\boldsymbol{1}^\top\boldsymbol{d}_t+e_t - g_t, c_t\} \label{DCevolution}\\
    & \boldsymbol{0} \preceq \boldsymbol{d}_t\preceq \boldsymbol{\overline{d}}.
\end{align}
\end{subequations}
Note: (\ref{SoCevolution}) and (\ref{DCevolution}) introduce strong temporal dependencies.

The value function $V_t^\pi(\cdot)$ is defined as the total reward received under the policy $\pi$, i.e., $V_t^\pi(x) = \mathbb E[\sum_{t=0}^T r_t(x_t, \pi(x_t))|x_t=x]$. The optimal value function is defined as $V_t^*(x)=\max_\pi V_t^\pi(x)$.

\section{Dynamic Program Analytical Results}
\label{sec:method}
In this section, we show that the optimal policy for the co-optimization problem in \autoref{sec:problem} is a threshold policy.

The following lemma establishes a conclusion regarding the concavity of the optimal value function.
\begin{lemma}\label{lem: concave}
    The optimal value function $V^*_t(s_t, g_t, c_t)$ is concave in $(s_t, c_t)$ for every time step $t\le T$. 
\end{lemma}
\begin{proof}
We follow the proof technique for lemma 1 in \cite{jin_optimal_2017}. 

Consider two different states $x_\tau$ and $x_\tau'$ with their corresponding optimal trajectories $\{(x_t, u_t)\}_{t=\tau}^T$ and $\{(x_t', u_t')\}_{t=\tau}^T$. Let $\tilde x_\tau=(x_\tau+x_\tau')/2$ be their average initial state, and define the average policy $\tilde \pi=\{\tilde u_t=(u_t+u_t')/2\}_{t=\tau}^T$ with corresponding trajectory $\{\tilde x_t\}_{t=\tau}^T$.

First, by the concavity of $U_t$:
\begin{equation}
    (U_t(\boldsymbol{d}_t) + U_t(\boldsymbol{d}_t'))/2\le U_t((\boldsymbol{d}_t+\boldsymbol{d}_t')/{2})=U_t(\tilde{\boldsymbol{d}}_t).\label{eq: u}
\end{equation}
The cost $P_t(z_t, c_t)$ consists of two components: the energy cost $p_t^+[z_t]^+-p_t^-[z_t]^-$ and the demand charging $p[z_t-c_t]^+$. Since $p_t^+\ge p_t^-\ge0$, the energy cost is convex. Then, 
\begin{equation}
\begin{aligned}
    &-\sum_{t=\tau}^T(p_t^+[z_t]^+-p_t^-[z_t]^-+p_t^+[z_t']^+-p_t^-[z_t']^-)/ 2\\
    &\le -\sum_{t=\tau}^T p_t^+[(z_t + z_t')/2]^+-p_t^-[(z_t+z_t')/2]^-\\
    &=-\sum_{t=\tau}^Tp_t^+[\tilde z_t]^+-p_t^-[\tilde z_t]^-.\label{eq: energy}
\end{aligned}
\end{equation}
And the demand charging cost 
\begin{equation}
    \begin{aligned}
        &-\sum_{t=\tau}^T(p[z_t-c_t]^++ p[z_t'-c_t']^+)/2 \\
        =&-p(\max_{\tau\le t\le T}\{z_t\}-c_\tau+\max_{\tau\le t\le T}\{z_t'\}-c_\tau')/2\\
        \le&-p(\max_{\tau\le t\le T}\{(z_t + z_t')/2\} - (c_\tau+c_\tau')/2)\\
        =&-p(\max_{\tau\le t\le T}\{\tilde z_t\} - \tilde c_\tau)=\sum_{t=\tau}^Tp[\tilde z_t - \tilde c_t]^+.
    \end{aligned}\label{eq: charge}
\end{equation}
Combining (\ref{eq: u}), (\ref{eq: energy}), (\ref{eq: charge}), 
\begin{equation}
    V^*_\tau(\tilde x_\tau)\ge V_\tau^{\tilde\pi}(\tilde x_\tau)\ge (V_\tau^*(x_\tau) + V_\tau^*(x_\tau'))/2,
\end{equation}
which proves the concavity of the optimal value function. 
\end{proof}

It turns out that the optimal value function is also non-decreasing in the system state.
\begin{lemma}
    The optimal value function $V_t^*(s_t,g_t,c_t)$ is non-decreasing in $s_t$, $g_t$, and $c_t$ for every time step $t\le T$.
\end{lemma}
\begin{proof}
    (1) $s_t$: Consider two SoCs $s_t$, $s_t'$ such that $s_t=s_t'+\epsilon$, $0<\epsilon\ll \underline{e}$. Let $\{u_\tau'=(e_\tau', \boldsymbol{d}_\tau')\}_{\tau=t}^T$ be the optimal action sequence for SoC $s_t'$. If there exists $t\le \tau_0< T$, $e_{\tau_0}'\ge -\underline{e}+\epsilon$, action sequence $\{u_\tau=(e_\tau'-\epsilon\cdot\mathbf 1[\tau=\tau_0], \boldsymbol{d}_\tau')\}_{\tau=t}^T$ is feasible for SoC $s_t$. Hence $V^*_t(s_t, g_t, c_t)-V^*_t(s_t',g_t,c_t)\ge P_{\tau_0}(z_{\tau_0}', c_{\tau_0}')-P_{\tau_0}(z_{\tau_0}'-\epsilon, c_{\tau_0}')\ge 0$. Otherwise, $e_\tau'<-\underline{e}+\epsilon<0$, $t\le\tau\le T$, so action sequence $\{u_\tau'\}_{\tau=t}^T$ is feasible for SoC $s_t$, $V^*_t(s_t, g_t, c_t)-V^*_t(s_t',g_t,c_t)\ge \gamma(s_T-s_T')=\gamma\epsilon>0$.    
    
    (2) $g_t$: Consider two generations $g_t=g_t'+\epsilon$, $\epsilon>0$. Let $\pi'=\{u_\tau'=(e_\tau', \boldsymbol{d}_\tau')\}_{\tau=t}^T$ be the optimal action sequence for generation $g_t'$. Apply the action sequence to generation $g_t$, $z_{t+1}=e_t'+\mathbf{1}^\top\boldsymbol{d}_t'-g_t=z_{t+1}'-\epsilon$, $V_t^*(s_t, g_t, c_t)-V_t^*(s_t, g_t', c_t)\ge P_t(z_t', c_t) - P_t(z_t'-\epsilon, c_t)\ge 0$.     
    
    (3) $c_t$: From (\ref{eq:payment}), $P_t(z_t,c_t)$ is non-increasing in $c_t$, hence $V_t^*(s_t,g_t, c_t)$ is non-decreasing in $c_t$.
\end{proof}

Lastly, we derive a property that links the optimal Q-function, defined as $Q_t^*(x_t, u_t) = r(x_t, u_t) + \mathbb E[V_{t+1}^*(x_{t+1})]$, to the system's control action $u_t$.
\begin{lemma}\label{lem:Qfunction}
The optimal Q-function $Q_t^*(x_t, u_t)$ is concave in $u_t$.     
\end{lemma}
\begin{proof}
For $u_\tau\neq u_\tau'$, let $\{(x_t, u_t)\}_{t=\tau}^T$ be the trajectory of $Q_\tau^*(x_\tau, u_\tau)$ and $\{(x_t', u_t')\}_{t=\tau}^T$ be the trajectory of $Q_\tau^*(x_\tau, u_\tau')$. By (\ref{eq: u}) - (\ref{eq: charge}), applying actions $\{(u_t+u_t')/2\}_{t=\tau}^T$ to $x_\tau$ gives:
\begin{equation}
    \begin{aligned}
        &Q_\tau^*(x_\tau, u_\tau)+Q_\tau^*(x_\tau, u_\tau')\\
        \le&\mathbb E\left[\sum_{t=\tau}^Tr_t(x_t, (u_t+u_t')/2)|x_\tau\right]\\
        \le& Q_\tau^*(x_\tau, (u_\tau+u_\tau')/2).
    \end{aligned}
\end{equation}
    
Therefore, $Q_t^*(x_t, u_t)$ is concave in $u_t$.
\end{proof}

\autoref{lem: concave} to \autoref{lem:Qfunction} lead to \autoref{thm: opt} about the structure of the optimal co-optimization policy under demand charges.
\begin{theorem}\label{thm: opt}
    For each step $t=0, 1, \cdots, T-1$, the optimal policy is a threshold policy. 
\end{theorem}
\begin{proof}
    The feasible region for action $u_t$, $(e_t, \boldsymbol{d}_t)\in \mathcal F_t= [\max\{-\underline{e}, -s_t\}, \min\{\overline{e}, B-s_t\}]\times\prod_{i=1}^K[0, d_i]$ is a compact set. From \autoref{lem: concave} to \autoref{lem:Qfunction}, the optimal value function and Q-function is concave in action $u_t$, so the optimal action
    \begin{equation}
        u_t^*=\arg\max_{u\in\mathcal F_t} r_t(x_t, u_t) + Q_t^*(x_t, u_t)
    \end{equation}
    must exist in the interior of the feasible set $\mathcal F_t$ or at the boundary of the feasible set $\mathcal F_t$. The optimal action can be obtained by solving concave optimizations with and without boundary constraints, which results in a threshold policy.
\end{proof}
Although \autoref{thm: opt} reveals a structured policy, direct dynamic programming implementation remains impractical. First, optimization necessitates gradients of $r_t$ and $Q_t$, which are computationally prohibitive for complex utilities. Also, continuous-space Q-function tracking requires costly discretization, as validated in simpler settings~\cite{jin_optimal_2017}. We thus employ reinforcement learning (RL) to approximate the optimal policy efficiently.  
\section{Case Study}
\label{sec:case_study}
We investigate the effectiveness and efficiency of RL in solving this co-optimization problem. We compare Proximal Policy Optimization (PPO)~\cite{schulman2017proximal} against Tesla's backup mode~\cite{tesla_powerwall_modes}, a threshold algorithm, and the optimal nonlinear program in simulation. Experiments were run on a server with 1 AMD Ryzen Threadripper 3990X 64-Core Processor and 4 Nvidia RTX A4000 GPUs.

\subsection{Problem Setup}
We adopt the concave utility function from \cite{alahmed_integrating_2022}: $U_t(d_t)=\alpha_t d_t - \frac{1}{2}\beta_td_t^2$, where $\alpha_t$, $\beta_t$ dynamically reflect the electricity costs and demand elasticity. The elasticity is set to -0.1 for an inelastic user~\cite{asadinejad2018evaluation}. Baseline battery parameters are set as $B=5$ kWh, $\underline{e}=\overline{e}=1$ kW with efficiencies $\tau=\rho=0.95$. Electricity price is assumed static of $\$0.12/\text{kWh}$ for buying, $\$0.06/\text{kWh}$ for selling, and $\$10/\text{kW}$ for daily peak demand charge.


\subsection{Control Strategies}

\paragraph{Baseline}
The Tesla's backup mode maintains a full battery by charging and only discharges during grid outages. It shows what would happen without a battery system~\cite{jin_optimal_2017}. 


\paragraph{Reinforcement Learning using PPO}
The PPO algorithm uses an actor network and a critic network with the same Multilayer Perceptron (MLP) structure. For this case study, we used the default MLP policy from stable baseline~\cite{stable-baselines3}.

\paragraph{Threshold Algorithm}
Based on the optimal structure, we developed a simple threshold algorithm as follows:
\begin{equation}
    e_t=\begin{cases}
        \max\{-\underline{e}, \rho s_t, -(\bar{\boldsymbol{d}}_t - g_t)\}, & \bar{\boldsymbol{d}}_t - g_t>0\\
        \min\{\overline{e}, (B-s_t)/\tau, -(\bar{\boldsymbol{d}}_t - g_t)\} & \bar{\boldsymbol{d}}_t - g_t\le 0
    \end{cases}
\end{equation}
and $\boldsymbol{d}_t=\bar{\boldsymbol{d}}_t$ in both cases.

\paragraph{Nonlinear Program}
Formulate the battery-demand co-optimization as a deterministic nonlinear program solved with Gurobi. It provides the optimal value for comparison.

\subsection{Dataset}
\begin{figure}[t]
    \centering
    \includegraphics[width=0.95\linewidth]{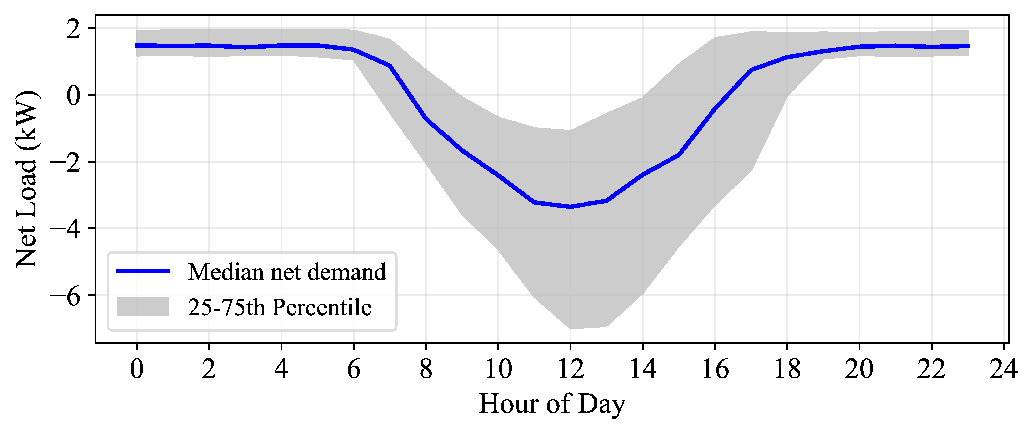}
    \caption{Daily net demand data from June 2022 to May 2024}
    \label{fig:net_demand}
\end{figure}
The Harvard HouseZero dataset \cite{han_two-year_2024} provides hourly demand and solar generation (June 2022--May 2024). \autoref{fig:net_demand} illustrates the net demand distribution. The PPO agent was trained on daily data (June 2022--April 2024). All test cases are generated from May 2024 daily data.  


\subsection{Results}

\begin{table*}[t]
\centering
\setlength{\tabcolsep}{6pt}
\caption{Reward (surplus percentage gain) of different methods under three different demand and generation scenarios}
\label{table:result_compare_demand_gen}
\begin{tabular}{ccccccccc}
\toprule
\multicolumn{4}{c}{Scenario} & \multicolumn{4}{c}{Method} \\
\cmidrule(r){1-4}  \cmidrule(r){5-8}
Generation & Demand & Battery Capacity & Charge/Discharge Limits & Backup & Threshold & PPO & Optimal \\ 
\midrule
25th Percentile & 75th Percentile & 5kWh & 1kW & 8.7 (0\%) & 10.3 (18.4\%) & 15.7 (80.5\%) & 20.3 (133.3\%) \\
50th Percentile & 50th Percentile & 5kWh & 1kW & 10.4 (0\%) & 11.9 (14.4\%) & 17.2 (65.4\%) & 21.9 (110.6\%) \\
75th Percentile & 25th Percentile & 5kWh & 1kW & 10.7 (0\%)& 12.2 (14.0\%) & 17.0 (58.9\%) & 22.0 (105.6\%) \\
\midrule
50th Percentile & 50th Percentile & 3kWh & 1kW & 10.4 (0\%) & 11.2 (7.7\%) & 14.8 (42.3\%) & 19.8 (90.4\%) \\
50th Percentile & 50th Percentile & 7kWh & 1kW & 10.4 (0\%)& 21.5 (106.7\%) & 17.9 (72.1\%) & 22.5(116.3\%) \\
\midrule
50th Percentile & 50th Percentile & 5kWh & 0.5kW & 10.4 (0\%)& 16.0 (53.8\%) & 16.2 (55.8\%) & 20.2 (94.2\%) \\
50th Percentile & 50th Percentile & 5kWh & 2kW & 10.4 (0\%)& 11.9 (14.3\%) & 18.5 (77.9\%) & 21.9 (110.6\%) \\
\bottomrule
\end{tabular}%

\end{table*}


\begin{figure*}[th]
    \centering
    \subfigure[Battery SoC comparison]{
    \includegraphics[width=0.31\linewidth]{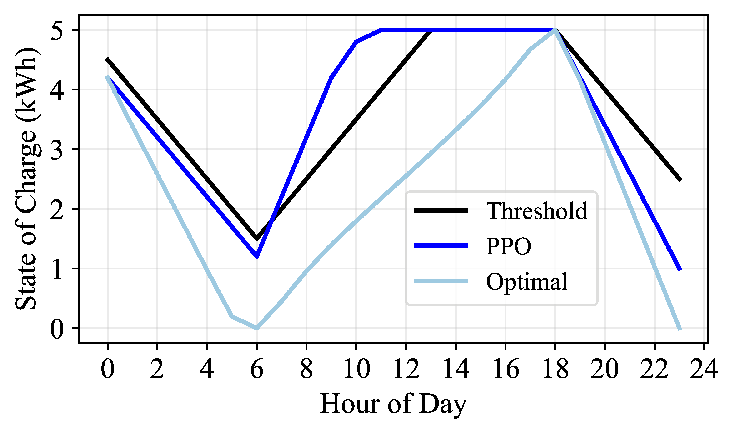}
    }
    \subfigure[Battery action comparison]{
    \includegraphics[width=0.31\linewidth]{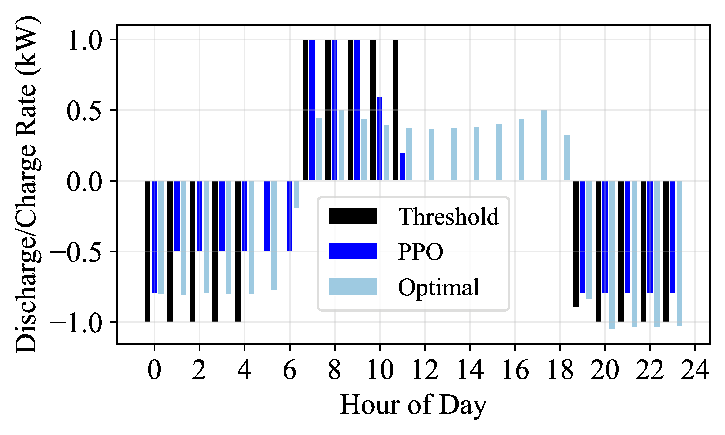}
    }
    \subfigure[User demand action comparison]{\includegraphics[width=0.31\linewidth]{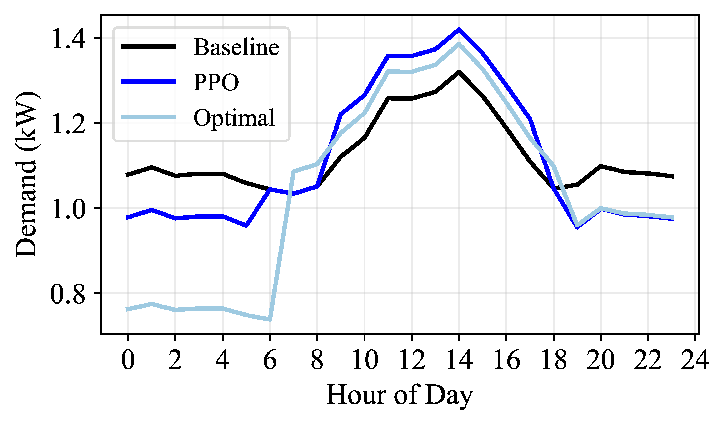}
    }
    \caption{Typical day battery state of charge and demand comparison between different control methods}
    \label{fig:medium}
\end{figure*}


We evaluate seven scenarios by varying:
\begin{itemize}
    \item \textbf{Generation/Demand}: 25th–75th percentile solar generation paired with 75th–25th percentile demand to simulate energy-abundant/scarce conditions;
    \item \textbf{Battery Capacity}: 3kWh - 7kWh capacity to simulate small/medium/large battery size.
    \item \textbf{Battery Charge/Discharge Limit}: 0.5 kW - 2kW to simulate low charging/discharging rates to high rates.
\end{itemize}

\autoref{table:result_compare_demand_gen} shows the performance comparisons across all scenarios. The results are evaluated using reward metrics, with percentage gains calculated relative to the Backup baseline. Under varying generation-demand scenarios with fixed battery parameters (5kWh, 1kW), PPO dominates the threshold method with 58.9-80.5\% versus 14.0-18.4\% gains, demonstrating adaptability to varying generation-demand conditions. The threshold algorithm is sensitive to the capacity, which is improved significantly (106.7\% gain) due to reduced demand-side dependency when the capacity is large, where the optimal and threshold policies are similar. By contrast, PPO maintains the robustness. Higher limits also amplify potential gains. While PPO follows the trend, the threshold shows over-aggressive control, degrading performance significantly. 

\autoref{fig:medium} shows the state and actions sequence of medium generation/demand and (5kWh, 1kW) battery configuration. Both threshold and RL policies align closely with optimal strategies in the scenario while being more aggressive in battery actions when lacking generation predictions. Key behavioral patterns include battery charging and demand elevation in the daytime to exploit solar surplus, and battery discharging and demand reduction in the evening to mitigate grid reliance.

\section{Conclusion}
\label{sec:conclusion}
In conclusion, we presented a comprehensive investigation into the co-optimization of DERs in a realistic setting controlling both demands and energy storage under a bi-directional net metering pricing structure, including demand charge. We proposed a stochastic dynamic program to model the problem and proved the optimality of the threshold structure. Through simulations with real-world data, we demonstrated that the RL algorithm can achieve significant surplus gains compared to baseline approaches, highlighting the effectiveness of using RL to approximate optimal control for prosumer EMS.

\bibliographystyle{IEEEtran}
\bibliography{ref}


\end{document}